\documentclass{article}
\usepackage{amsfonts,amssymb,amsthm,amsmath}
\usepackage{graphicx}
\usepackage{footnote}
\usepackage{float}
\makesavenoteenv{tabular}
\usepackage{comment}
\usepackage[bookmarks]{hyperref}
\usepackage{subcaption}
\usepackage{authblk}

\theoremstyle{plain}
\newtheorem{thm}{Theorem}

\newtheorem{prop}[thm]{Proposition}

\theoremstyle{definition}
\newtheorem{deph}{Definition}

\DeclareMathOperator{\E}{\mathbb{E}}

\DeclareMathOperator{\m}{\mathcal{M}}
\DeclareMathOperator{\lin}{\mathcal{L}}
\DeclareMathOperator{\e}{\mathcal{E}}
\DeclareMathOperator{\f}{\mathcal{F}}
\DeclareMathOperator{\s}{\sigma}
\DeclareMathOperator{\modu}{mod}

\DeclareMathOperator{\D}{\partial}
\DeclareMathOperator*{\argmin}{arg\,min}

\makeatletter
\renewcommand{\maketag@@@}[1]{\hbox{\m@th\normalsize\normalfont#1}}%
\makeatother

\title{\sc Iterative Scaling Algorithm for Channels}
\author[1]{Paolo Perrone\thanks{Correspondence: perrone@mis.mpg.de}}
\author[1,2,3]{Nihat Ay}
\affil[1]{\small Max Planck Institute for Mathematics in the Sciences\\ Inselstrasse 22\\ 04103 Leipzig, Germany}
\affil[2]{\small Faculty of Mathematics and Computer Science, University of Leipzig\\ PF 100920\\ 04109 Leipzig, Germany}
\affil[3]{\small Santa Fe Institute, 1399 Hyde Park Road, Santa Fe, NM 87501, USA}
\date{}
\begin{document}
\maketitle
\thispagestyle{empty}
\addcontentsline{toc}{subsection}{Abstract}

\begin{abstract}

Here we define a procedure for evaluating KL-projections (I- and rI-projections) of channels. These can be useful in the decomposition of mutual information 
between input and outputs, e.g. to quantify synergies and interactions of different orders, as well as information integration and 
other related measures of complexity.

The algorithm is a generalization of the standard iterative scaling algorithm, which we here extend from probability distributions to channels
(also known as transition kernels). 

\end{abstract}

\section{Introduction}

Here we present an algorithm to compute projections of channels onto exponential families of fixed interactions.

The decomposition is geometrical, and it is based on the idea that, rather than joint distributions, the 
quantities we work with are channels, or conditionals
(or Markov kernels, stochastic kernels, transition kernels, stochastic maps).
Our algorithm can be considered a channel version of the iterative scaling of (joint) probability distributions,
presented in \cite{csiszar}.  


Exponential and mixture families (of joints and of channels) have a duality property, 
shown in Section \ref{fam}. By fixing some marginals,
one determines a mixture family. By fixing (Boltzmann-type) interactions, one determines an exponential family.
These two families intersect in a single point, which means that (Theorem \ref{dualmk}) 
\emph{there exists a unique element which has the desired marginals and the desired interactions}.

As a consequence, Theorem \ref{dualmk} translates projections onto exponential families (which are generally hard to compute)
to projections onto fixed-marginals mixture families (which can be approximated by an iterative procedure). 
Section \ref{algo} explains how this is done.

Projections onto exponential families are becoming more and more important in the definition of measures 
of statistical interaction, complexity, synergy, and related quantities.
In particular, the algorithm can be used to compute decompositions of mutual information, as for example the ones defined 
in \cite{olbrich} and \cite{us}, and it was indeed used to compute all the numerical examples
in \cite{us}.
Another application of the algorithm is explicit computations of complexity measure as treated in 
\cite{preprint2}, \cite{amaripreprint}, and \cite{amaribook}.
Examples of both applications can be found in Section \ref{applic}.

For all the technical details about the iterative scaling algorithm
in its traditional version, we refer the interested reader to Chapters 3 and 5 of \cite{csiszar}.

All proofs can be found in the Appendix.

\subsection{Technical Definitions}

We take the same definitions and notations as in \cite{us}, except that we let the output be multiple.
More precisely, we consider a set of $N$ input nodes $V$, taking values in the sets $X_1,\dots,X_N$, and a set of $M$ output nodes $W$,
taking values in the sets $Y_1,\dots,Y_M$. We write the input globally as $X:=X_1\times\dots\times X_N$, and the output globally as $Y:=Y_1\times\dots\times Y_M$.
We denote by $F(Y)$ the set of real functions on $Y$, and by $P(X)$ the set of probability measures on $X$. 

\begin{deph}
 Let $I\subseteq V$ and $J\subseteq W$. 
 We call $F_{IJ}$ the space of functions which only depend on $X_I$ and $Y_J$:
 \begin{align}
  F_{IJ} := \big\{ f &\in F(X,Y)\;\big| \notag \\ 
  &f(x_I,x_{I^c},y_J,y_{J^c})=f(x_I,x'_{I^c},y_J,y'_{J^c})\;\forall x_{I^c}, x'_{I^c},y_{J^c},y'_{J^c} \big\}\;.
 \end{align}
\end{deph}

We can model the channel from $X$ to $Y$ as a Markov kernel $k$, that assigns to each $x\in X$ a probability measure on 
$Y$ (for a detailed treatment, see \cite{kakihara}).
Here we will consider only finite systems, so we can think of a channel simply
as a transition matrix (or stochastic matrix), whose rows sum to one.
\begin{equation}\label{stoc}
 k(x;y)\ge 0\quad \forall x,y; \qquad \sum_{y} k(x;y) =1 \quad \forall x\;.
\end{equation}
The space of channels from $X$ to $Y$ will be denoted by $K(X;Y)$.
We will denote by $X$ and $Y$ also the corresponding random variables, whenever this does not lead to confusion.

Conditional probabilities define channels: if $p(X,Y)\in P(X,Y)$ and the marginal $p(X)$ is strictly positive, then $p(Y|X)\in K(X;Y)$ is a well-defined
channel. Viceversa, if $k\in K(X;Y)$, given $p\in P(X)$ we can form a well-defined joint probability:
\begin{equation}
 pk(x,y):= p(x)\,k(x;y)\quad \forall x,y\;.
\end{equation}

To extend the notion of divergence from probability distributions to channels, we need an ``input distribution'':
\begin{deph}
 Let $p\in P(X)$, let $k,m\in K(X;Y)$. Then:
\begin{equation}\label{mkdiv}
 D_p(k||m):= \sum_{x,y} p(x)\,k(x;y)\,\log\dfrac{k(x;y)}{m(x;y)}\;.
\end{equation}
\end{deph}

Let $p,q$ be joint probability distributions on $X\times Y$, and let $D$ be the KL-divergence. Then the following ``chain rule'' holds:
 \begin{equation}\label{pkdiv}
  D(p(X,Y)||q(X,Y)) = D(p(X)||q(X)) + D_{p(X)}(p(Y|X)||q(Y|X))\;.
 \end{equation}

\section{Families of Channels}\label{fam}

Suppose we have a family $\e$ of channels, and a channel $k$ that may not be in $\e$. Then we
can define the ``distance'' between $k$ and $\e$ in terms of $D_p$.
\begin{deph}
 Let $p$ be an input distribution. The divergence between a channel $k$ and a family of channels $\e$ is given by:
 \begin{equation}
  D_p(k||\e):=\inf_{m\in\e} D_p(k||m)\;.
 \end{equation}
 If the minimum is uniquely realized, we call the channel
 \begin{equation}
  \pi_{\e}k:=\argmin_{m\in\e} D_p(k||m)\;
 \end{equation}
 the \emph{rI-projection} of $k$ on $\e$ (and simply ``an'' rI-projection if it is not unique). 
\end{deph}

The families considered here are of two types, dual to each other: linear and exponential. For both cases, we take
the closures, so that the minima defined above always exist. 
\begin{deph}
 A \emph{mixture family} of $K(X;Y)$ is a subset of $K(X;Y)$ defined by one or several affine equations, i.e., the locus of the $k$
 which satisfy a (finite) system of equations in the form:
 \begin{equation}\label{mixture}
  \sum_{x,y} k(x;y) f_i(x,y) = c_i\;,
 \end{equation}
 for some functions $f_i\in F(X,Y)$, and some constants $c_i$.
\end{deph}

\paragraph{Example.} Consider a channel $m\in K(X;Y_1,Y_2)$. 
We can form the marginal:
\begin{equation}
 m(x;y_1):= \sum_{y_2}m(x;y_1,y_2)\;.
\end{equation}
The channels $k\in K(X;Y_1,Y_2)$ such that $k(x;y_1)=m(x;y_1)$ form a mixture family,
defined by the system of equations (for all $x'\in X$, $y'_1\in Y_1$):
\begin{equation}
 \sum_{x,y_1,y_2}k(x;y_1,y_2)\,\delta(x,x')\delta(y_1,y'_1) = m(x';y'_1)\;,
\end{equation}
where the function $\delta(z,z')$ is equal to 1 for $z=z'$, and zero for any other case.

More in general, let $\lin$ be a (finite-dimensional) linear subspace of $F(X,Y)$, and let $k\in K(X;Y)$. Then:
 \begin{equation}\label{m}
  \m(k,\lin) := \bigg\{ m\in K(X;Y)\,\bigg| \sum_{x,y}m(x;y)l(x,y)=\sum_{x,y}k(x;y)l(x,y) \; \forall l\in \lin \bigg\}\,
 \end{equation}
 is a mixture family, which we call \emph{generated by $k$ and $\lin$}.
 
\begin{deph}
 A (closed) \emph{exponential family} of $K(X;Y)$ is (the closure of) a subset of $K(X;Y)$ of channels in the form:
 \begin{equation}
  \dfrac{e^{f(x,y)}}{Z(x)}\, k(x;y) \;,
 \end{equation}
 where $f$ satisfies affine constraints, $k$ is fixed, and:
 \begin{equation}
  Z(x):=\sum_y e^{f(x,y)}\,k(x;y)\;
 \end{equation}
 so that the channel is correctly normalized.
\end{deph}

This is a sort of multiplicative equivalent of mixture families, as the exponent satisfies constraints similar to \eqref{mixture}.

\paragraph{Example.} 
 Let $\lin$ be a (finite-dimensional) linear subspace of $F(X,Y)$, and let $k\in K(X;Y)$. Then the closure:
 \begin{equation}\label{e}
  \e(k,\lin) := \bigg\{ \dfrac{e^{l(x,y)}}{Z(x)}\, k(x;y) \,\bigg|\,Z(x)=\sum_y e^{l(x,y)}\,k(x;y),\,l\in \lin \bigg\}\;
 \end{equation}
 is an exponential family, which again we call \emph{generated by $k$ and $\lin$}.

This family is in some sense ``dual'' to the family in \eqref{m}.
The duality is expressed more precisely by the following result.

\begin{thm}\label{dualmk}
 Let $\lin$ be a subspace of $F(X,Y)$. Let $p\in P(X)$ be strictly positive. 
 Let $k_0\in K(X;Y)$ be a strictly positive ``reference'' channel.
 Let $\e := \e(k_0,\lin)$ and $\m := \m(k,\lin)$. For $k'\in K(X;Y)$, the following conditions are equivalent:
 \begin{enumerate}
  \item $k'\in \m\cap\e$.
  \item $k'\in \e$, and $D_p(k||k')=\inf_{m\in\e} D_p(k||m)$.
  \item $k'\in \m$, and $D_p(k'||k_0)=\inf_{m\in\m} D_p(m||k_0)$.
 \end{enumerate}
 In particular, $k'$ is unique, and it is exactly $\pi_{\e}k$.
\end{thm}

Geometrically, we are saying that $k'=\pi_{\e} k$, the rI-projection of $k$ on $\e$.
We call the mapping $k\to k'$ the \emph{rI-projection operator}, and the mapping $k_0\to k'$ the \emph{I-projection operator}
These are the channel equivalent of the I-projections introduced in \cite{iproj}
and generalized in \cite{iriv}. The result is illustrated in Figure \ref{fig:dualmk}.

\begin{figure}[H]
 \centering
 \includegraphics[scale=1,keepaspectratio=true]{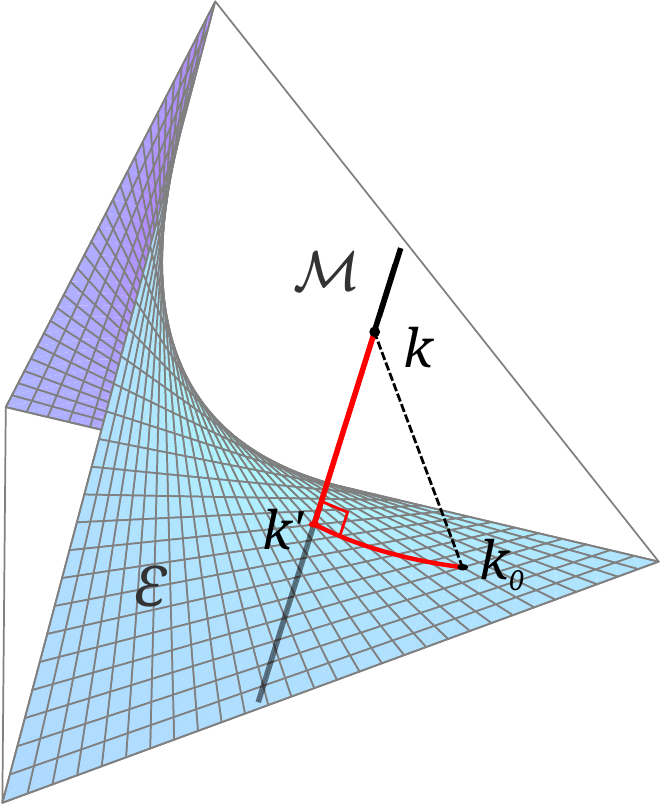}
 \caption{Illustration of Theorem \ref{dualmk}. The point $k'$ at the intersection minimizes on $\e$ the distance from $k$, and minimizes on $\m$ the distance from $k_0$.}
 \label{fig:dualmk}
\end{figure}

As suggested by Figure \ref{fig:dualmk}, I- and rI-projections on exponential families satisfy a Pythagoras-type equality.
For any $m\in \e$, with $\e$ exponential family:
\begin{equation}\label{py}
  D_p(k||m)= D_p(k||\pi_{\e}k) + D_p( \pi_{\e}k||m)\;.
\end{equation}
This statement follows directly from the analogous statement for probability distribution found in \cite{hierarchy},
after applying the chain rule \eqref{pkdiv}.

\section{Algorithm}\label{algo}

The algorithm can be considered as a channel equivalent of the iterative scaling procedure for joint distributions, which can be found in Chapter 5 of \cite{csiszar}.
Translated into our language, that theorem says the following:
\begin{thm}[Theorem 5.1 of \cite{csiszar}]\label{jointit}
Let $\lin_1,\dots,\lin_n$ be mixture families of joint distributions with nonempty intersection $\lin$.
Denote by $\Sigma_iq$ the $I$-projection of a joint $q$ onto the family $\lin_i$.
Consider the sequence that starts at $q^0$ and is defined iteratively by:
 \begin{equation}
  q^j := \Sigma_{(j \modu n)} q^{j-1}\;.
 \end{equation}

 Then $\{q^j\}$ converges, and the limit point is the $I$-projection of $q^0$ onto $\lin$, i.e.
 if we call:
 \begin{equation}
  \lim_{i\to\infty} q^i := q\;,
 \end{equation}
 then $q\in \lin$, and for any $\bar q\in \lin$:
 \begin{equation}
  D(q^0||\bar q) =  D(q^0||q) + D(q||\bar q)\;.
 \end{equation}
\end{thm}

Our result depends on the theorem above, in the following way. We define a marginal procedure for channels, 
which in general depends on the choice of an input distribution. We define mixture families of channels with 
fixed marginals in a way compatible with the equivalent for joints.
We then define scalings of channels, and prove that they give the desired result at the joint level.
This makes it possible to translate the statement of Theorem \ref{jointit} to an analogous statement for channels,
Theorem \ref{convergence}.

Unless otherwise stated, all the input distributions here will be assumed strictly positive.
All our proofs can be found in the appendix.

\begin{deph}
 Consider an input distribution $p\in P(X)$. Let $I\subseteq [N], J\subseteq [M], J\ne\emptyset$. We define the \emph{marginal operator} for channels as:
 \begin{equation}\label{margop}
  k(x;y)\mapsto k(x_I;y_J):= \sum_{x_{I^c},y_{J^c}}p(x_{I^c}|x_I)\,k(x_I,x_{I^c};y_J,y_{J^c})\;,
 \end{equation}
 given the input $p$. 
\end{deph}

\begin{prop}\label{jeq}
 Defined as above, $k(X_I;Y_J)$ is exactly the conditional probability for the marginal $pk(X_I,Y_J)$.
 In other words, $k(X,Y)$ has marginal $k(X_I;Y_J)$ if and only if $pk(X,Y)$ has marginal $pk(X_I,Y_J)$.
\end{prop}

\begin{deph}
 Consider an input distribution $p\in P(X)$.
 Let $I\subseteq [N], J\subseteq [M]$, $J\ne\emptyset$. We define the mixture families $\m_{IJ}(\bar{k})$ as:
 \begin{equation}
  \m_{IJ}(\bar{k}) := \big\{k(x_{1\dots n};y_{1\dots m}) \,\big|\, p(x_I)\,k(x_I;y_J) =  p(x_I)\, \bar k(x_I;y_J) \big\}\;,
 \end{equation}
 where the $\bar k(x_I;y_J)$ are prescribed channel marginals.
 
 Analogously, let $\bar q$ be a probability distribution in $P(X_I,Y_J)$. 
 We define the mixture families:
 \begin{equation}
  \mathcal{J}_{IJ}(\bar q) := \big\{q(x_{1\dots n},y_{1\dots m}) \,\big|\, q(x_I,y_J) = \bar q(x_I,y_J) \big\}\;.
 \end{equation}
\end{deph}

Proposition \ref{jeq} says that, for any $k\in K(X;Y)$, for any (strictly positive) $p\in P(X)$, and for any $I\subseteq [N], J\subseteq [M]$ :
\begin{equation}\label{corresp}
 k\in \m_{IJ}(\bar{k}) \quad \Leftrightarrow \quad pk \in  \mathcal{J}_{IJ}(p \bar k)\;.
\end{equation}

\begin{prop}\label{presc}
 $\m_{IJ}(\bar{k})$ is exactly the set $\m(\bar{k},\lin)$ of equation \eqref{m}, where as $\lin$
 we take the space $F_{IJ}$ of functions which only depend on the nodes in $I,J$.
\end{prop}

Just as in \cite{csiszar}, the I-projections for single marginals can be obtained by scaling. 
For joint distributions the scaling is done in this way: if $\bar p(X_I,Y_J)$ is a ``prescribed'' marginal, then:
\begin{equation}
 \s^{\bar q}_{IJ} p\,(X,Y) := p(X,Y)\,\dfrac{\bar q(X_I,Y_J)}{p(X_I,Y_J)}
\end{equation}
will have the prescribed marginals, and even be the I-projection of $p$ on $\mathcal{J}_{IJ}(\bar q)$. i.e., 
$\s^{\bar q}_{IJ} p \in\mathcal{J}_{IJ}(\bar q)$, and:
\begin{equation}
 D(p||\bar q) = D(p||\s^{\bar q}_{IJ} p ) + D(\s^{\bar q}_{IJ} p ||\bar q)\;.
\end{equation}
For the proof, see Chapter 3 and Section 5.1 of \cite{csiszar}.

In the case of channels, the scaling is instead done in two steps.

\begin{deph}
 We define the (unnormalized) \emph{$IJ$-scaling} as the operator $\s_{IJ}^{\bar k}:K(X;Y)\to F(X,Y)$, mapping $k$ to:
 \begin{equation}\label{usk}
  \s_{IJ}^{\bar k} k \, (x,y) := k(x_I , x_{I^c} ; y_J , y_{J^c} ) \frac{ \bar{k}(x_I ; y_J) }{ k(x_I ; y_J) }  \;.
 \end{equation}
\end{deph}

We have that $\s_{IJ}^{\bar k} k$ is \emph{not} an element of $ \m_{IJ}$, as in general it is not even in
$K(X;Y)$ (i.e. a correctly normalized channel). However, at the joint level this corresponds exactly to the joint scaling:

\begin{prop}\label{jlevel}
 Let $p\in P(X)$, $k\in K(X;Y)$, and $\bar k \in K(X_I;Y_J)$. Then:
 \begin{equation}
  p (\s_{IJ}^{\bar k} k) = \s_{IJ}^{p \bar k} pk\;.
 \end{equation}
\end{prop}
This implies that $p \s_{IJ}^{\bar k} k$ is the $I$-projection of $pk$ to the family $\mathcal{J}_{IJ} (p \bar k)$. 

\begin{deph}
 We define the \emph{normalized $IJ$-scaling} as the operator $N \s_{IJ}^{\bar k}:K(X,Y)\to K(X;Y)$, mapping $k$ to:
 \begin{equation}\label{sk}
  N \s_{IJ}^{\bar k} k \, (x,y) :=  \dfrac{1}{Z(x)}\, \s_{IJ}^{\bar k} k \, (x,y) \;,
 \end{equation}
  where:
 \begin{equation}
  Z(x) :=  \sum_{y'}\s_{IJ}^{\bar k} k \, (x,y')\;.
 \end{equation}
\end{deph}

At the joint level, this corresponds to scaling of the input in the following way:
\begin{prop}\label{inscale}
  Let $p\in P(X)$, $k\in K(X;Y)$, and $\bar k \in K(X_I;Y_J)$. Then:
 \begin{equation}
  p( N \s_{IJ}^{\bar k} k) = \s_{[N]}^{p} \s_{IJ}^{p \bar k} pk\;.
 \end{equation}
\end{prop}
This implies that $p N \s_{IJ}^{\bar k} k$ is the $I$-projection of $p\s_{IJ}^{\bar k} k$ to the family with prescribed input $p(X)$. 
For brevity, let's call this family $\mathcal{J}_{[N]}(p)$.

Now $N \s_{IJ}^{\bar k} k$ is an element of $K(X;Y)$, but still \emph{not} of $ \m_{IJ}$. 
However, if we iterate the operator $N \s_{IJ}^{\bar k}$, the resulting sequence will converge to the 
projector on $ \m_{IJ}$. More in general, the following result holds:

\begin{thm}[Main result]\label{convergence}
 For $1\le i \le n$, let  $I_i\subseteq [N]$ be subsets of $[N]$ and $J_i\subseteq [M]$ be nonempty 
 subsets of $[M]$. 
 Take an input distribution $p\in P(X)$ and a channel $\bar k\in K(X,Y)$.
 Define the mixture families of prescribed marginals:
 \begin{equation}
  \m_i = \m_{I_iJ_i}(\bar k(X_{I_i},Y_{J_i}))\;,
 \end{equation}
 and their intersection, which is also a mixture family (nonempty, as it contains at least $\bar k$):
 \begin{equation}
  \m := \bigcap_i \m_i\;.
 \end{equation}
 Choose a (different) channel $k^0\in K(X,Y)$ and consider the sequence of normalized scalings
 starting at $k^0$ and defined iteratively by:
 \begin{equation}
  k^j := N\sigma_{I_{(j \modu n)}J_{(j\modu n)}} k^{j-1}\;.
 \end{equation}
 Then:
 \begin{itemize}
  \item $k^i$ converges to a limit channel:
 \begin{equation}
  \lim_{i\to\infty} k^i := l\;;
 \end{equation}
 \item The limit $l$ is the $I$-projection of $k^0$ on $\m$, i.e. $l\in\m$ and:
 \begin{equation}
  D_p(k^0||\bar k) =  D_p(k^0||l) + D_p(l||\bar k)\;.
 \end{equation}
 \end{itemize}
\end{thm}

The proof can be found in the appendix.

To apply the Theorem \ref{convergence} in our algorithm, we choose as initial channel $k^0$ exactly the reference channel $k_0$ of Theorem \ref{dualmk},
usually the uniform channel. 
As $\bar k$ we take exactly the ``prescription channel'' $k$ of Theorem \ref{dualmk}, 
i.e. the channel which has the desired marginals.
The result of the iterative scaling will be the rI-projection of $k$ on the desired exponential family.

\section{Applications}\label{applic}

\subsection{Synergy Measures}

The algorithm presented here permits to compute the decompositions of mutual information between inputs and outputs
in \cite{olbrich} and \cite{us}. We give here examples of computations of \emph{pairwise synergy} as an rI-projection
for channels, as described 
in \cite{us}. It is not within the scope of this article to motivate this measure,
we rather want to show how it can be computed.

Let $k$ be a channel from $X=(X_1,X_2)$ to $Y$. Let $p\in P(X)$ be a strictly positive input distribution. 
We define in \cite{us} the synergy of $k$ as:
\begin{equation}\label{syn}
 d_2(k):= D_p(k||\e_1)\;,
\end{equation} 
where $\e_1$ is the (closure of the) family of channels in the form:
\begin{equation}
 m(x_1,x_2;y)=\dfrac{1}{Z(X)} \exp\big( \phi_0(x_1,x_2)+\phi_1(x_1,y)+\phi_2(x_2,y) \big)\;,
\end{equation}
where:
\begin{equation}
 Z(x):=\sum_y\exp\big( \phi_0(x_1,x_2)+\phi_1(x_1,y)+\phi_2(x_2,y) \big)\;,
\end{equation}
and:
\begin{equation}
 \phi_0 \in F_{\{1,2\}\emptyset}\,,\quad \phi_1 \in F_{\{1\}\{1\}}\,,\quad \phi_2 \in F_{\{2\}\{1\}}\;.
\end{equation}
According to Theorem \ref{dualmk}, the rI-projection of $k$ on $\e_1$ is the unique point $k'$ of $\e_1$ which has
all the prescribed marginals:
\begin{equation}
 k'(x_1;y) = k(x_1;y)\,,\quad k'(x_2;y) = k(x_2;y)\;,
\end{equation}
and can therefore be computed by iterative scaling, either of the joint distribution (as it is traditionally
done, see \cite{csiszar}), or of the channels (our algorithm). 

Here we present a comparison of the two algorithms, implemented similarly and in the same language (Mathematica).
The red dots represent our (channel) algorithm, and the blue dots represent the joint rescaling algorithm.

For the easiest channels (see Figure \ref{fig:xor}), both algorithm converge instantly.

\begin{figure}[H]
 \centering
 \includegraphics[scale=.5,keepaspectratio=true]{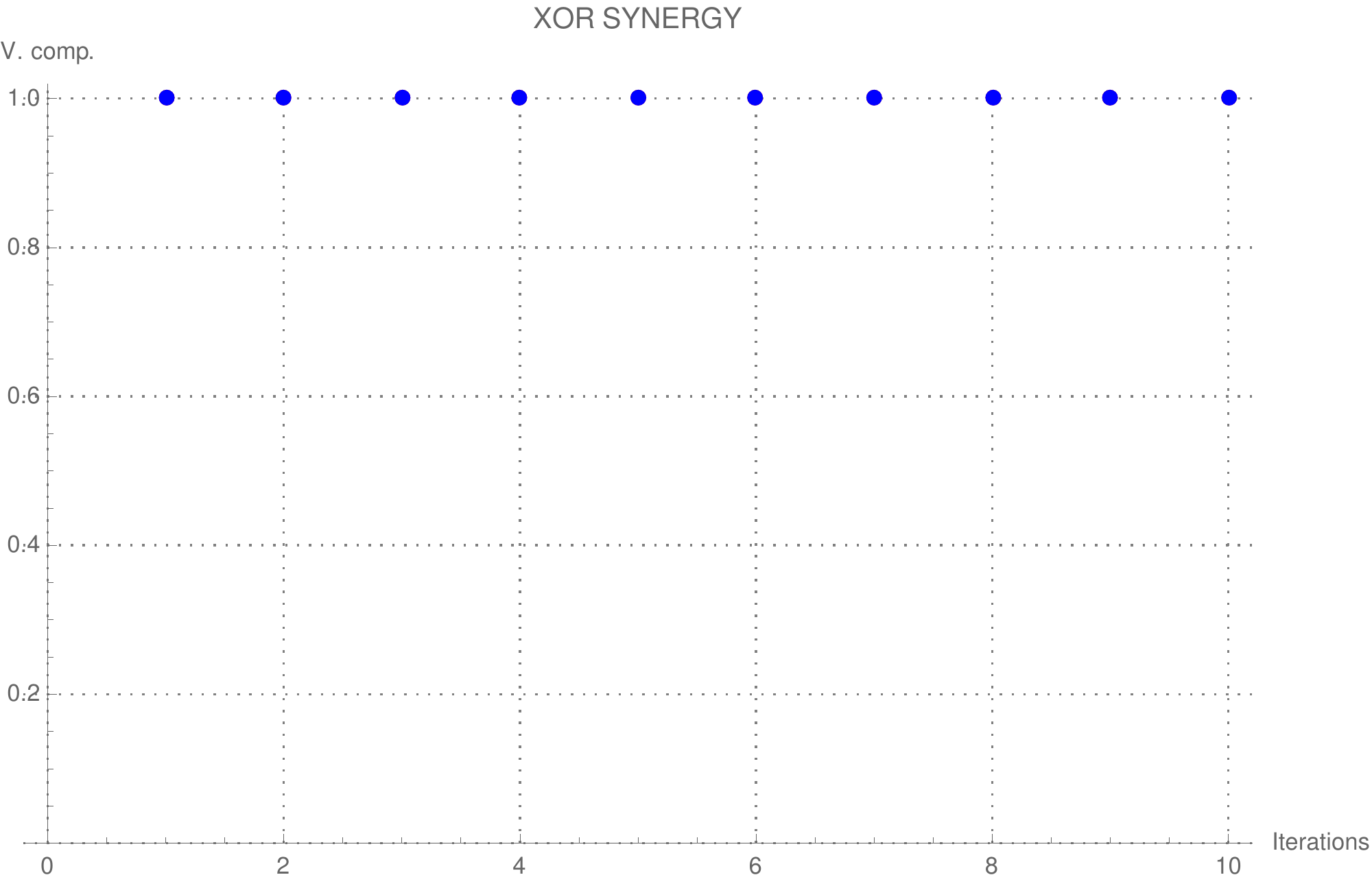}
 \caption{Comparison of convergence times for the synergy of the XOR gate. Both algorithms get immediatly the desired result.
 (The dots here are overlapping, the red ones are not visible.)}
 \label{fig:xor}
\end{figure}

A more interesting example is a randomly generated channel (Figure \ref{fig:rand}), in which both method need 5-10
iterations to get to the desired value. However, the channel method is slightly faster.

\begin{figure}[H]
 \centering
 \includegraphics[scale=.5,keepaspectratio=true]{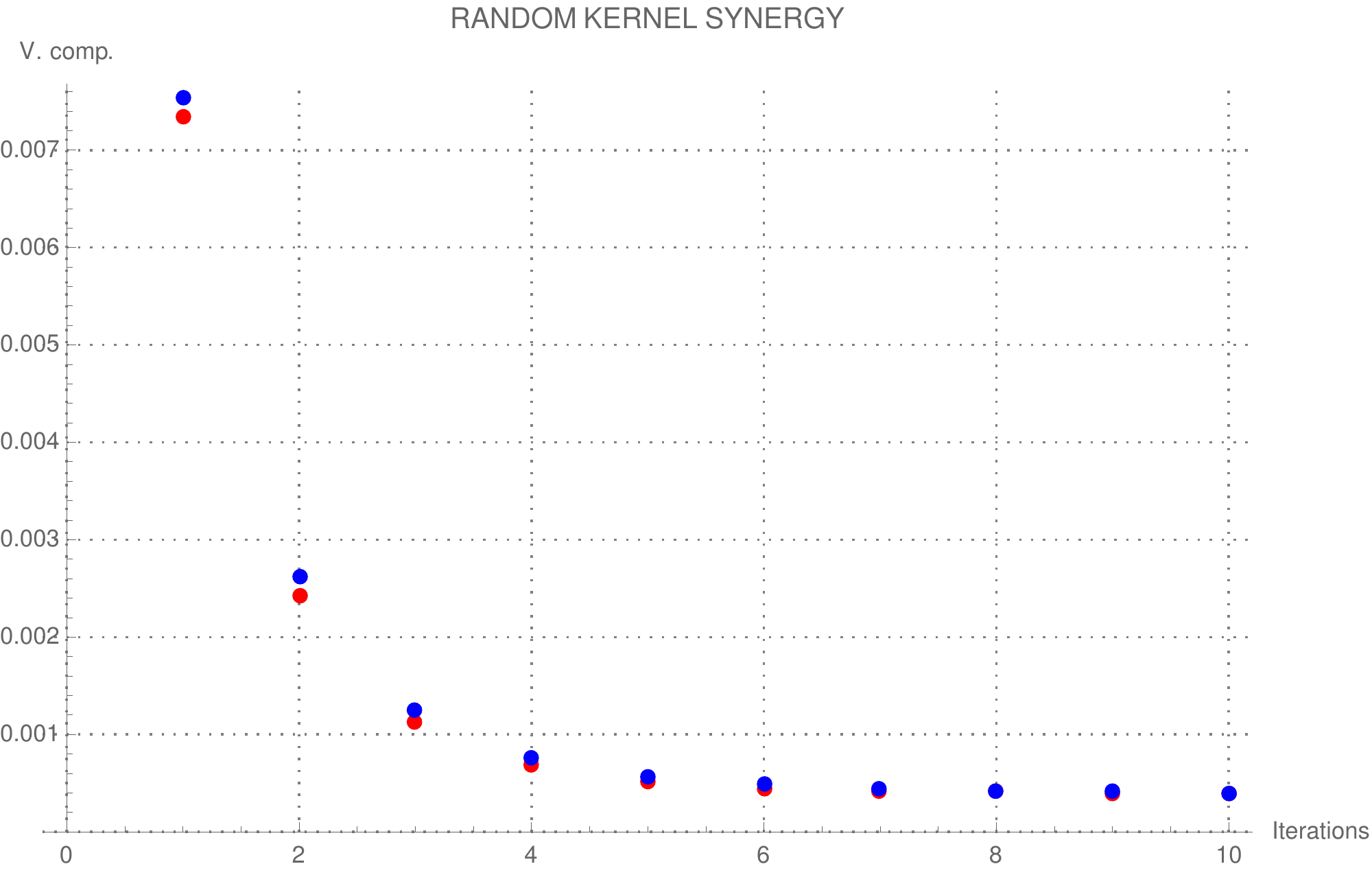}
 \caption{Comparison of convergence times for the synergy of a randomly generated channel. The channel method (red) is slightly faster.}
 \label{fig:rand}
\end{figure}

The most interesting example is the synergy of the AND gate, which should be zero according to the procedure \cite{us}. 
In that article, we mistakenly wrote a different value, that here we would like to correct (it is zero).
The convergence to zero is very slow, of the order of $1/n$ (Figure \ref{fig:and}). It is clearly again
slightly faster for the channel method in terms of iterations.

\begin{figure}[H]
 \centering
 \includegraphics[scale=.5,keepaspectratio=true]{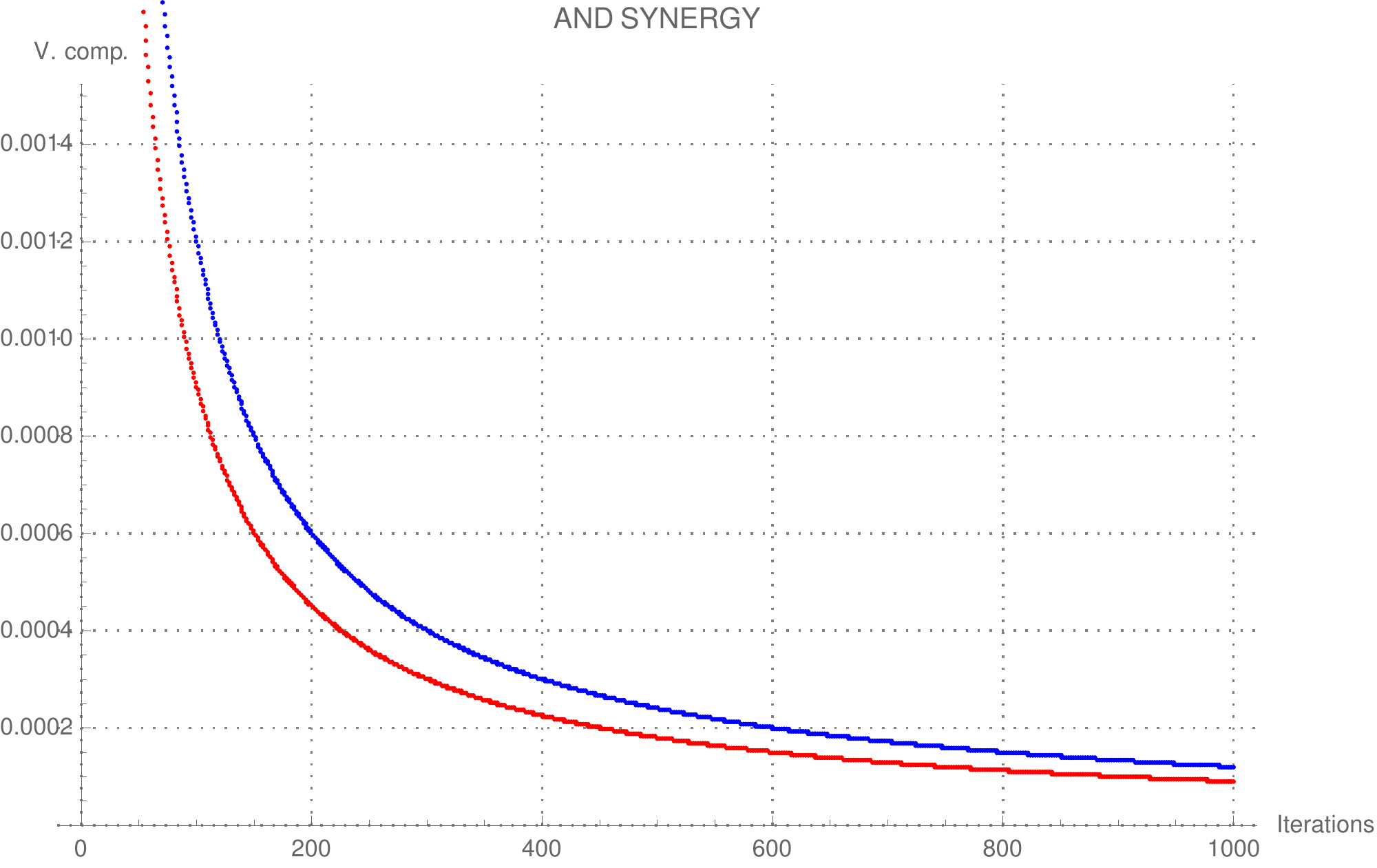}
 \caption{Comparison of convergence times for the synergy of the AND gate. The channel method (red) tends to zero 
 proportionally to $n^{-1.05}$, the joint method (blue) proportionally to $n^{-0.95}$.}
 \label{fig:and}
\end{figure}

It has to be noted, however, that rescaling a channel requires more elementary operations than rescaling
a joint distribution. Because of this, one single iteration with our method takes longer than with the joint
method. (As explained in Section \ref{algo}, a scaling for the channel corresponds to two scalings for the joint.) 
In the end, despite the need of fewer iterations, the total computation time of a projection
with our algorithm can be longer (depending on the particular problem). 
For example, again for the synergy of the AND gate, we can plot the computation time as a function of the 
accuracy (distance to actual value), down to $10^{-3}$. The results are shown in Figure \ref{fig:comp}.

\begin{figure}[H]
 \centering
 \includegraphics[scale=.5,keepaspectratio=true]{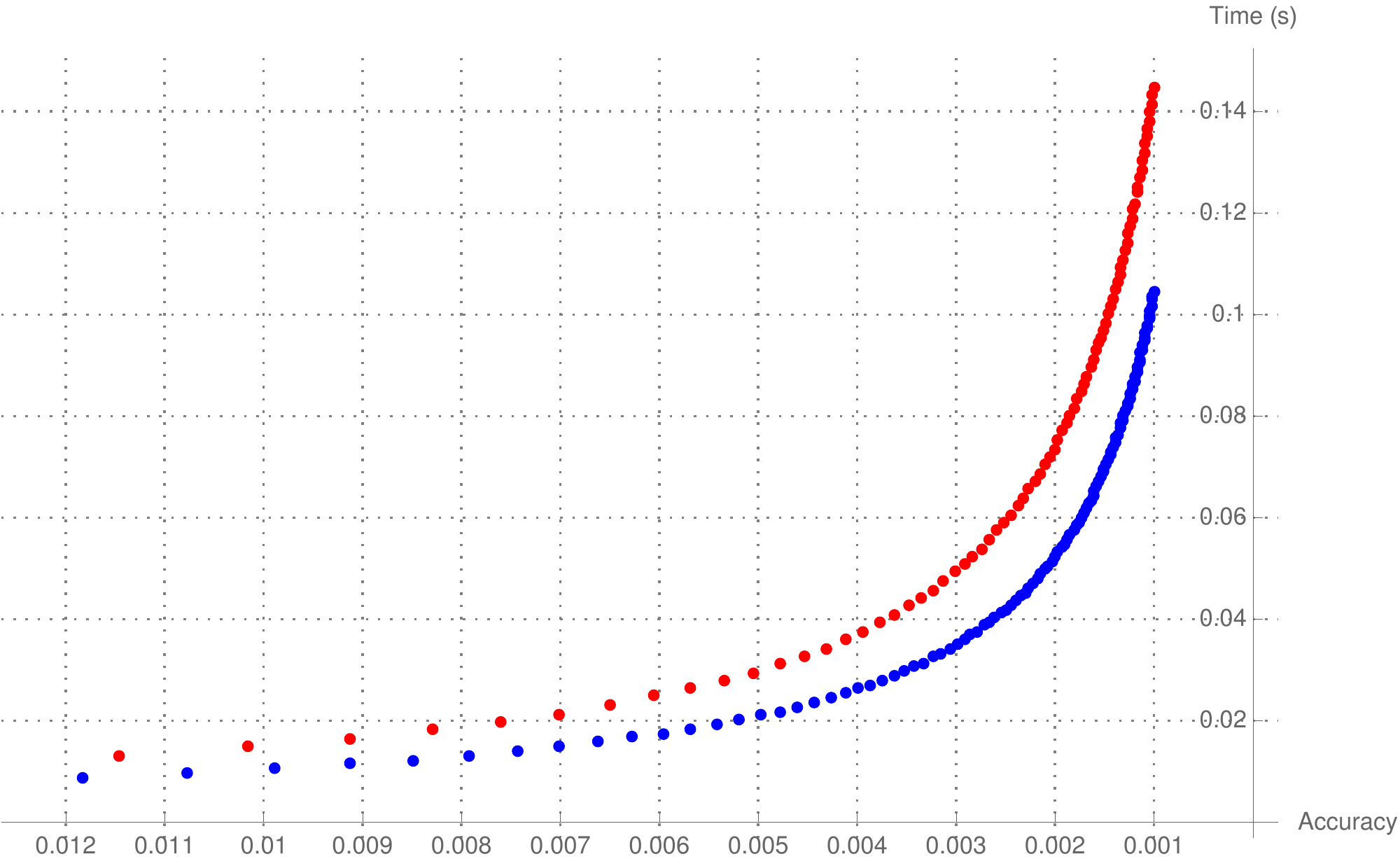}
 \caption{Comparison of total computation times for the synergy of the AND gate. The channel method (red) 
 is slightly slower than the joint method (blue).}
 \label{fig:comp}
\end{figure}

To get to the same accuracy, though, the channel approach used less iterations. 
In summary, our algorithm is better in terms of iteration complexity, 
but generally worse in terms of computing time.

\subsection{Complexity Measures}

Iterative scaling can also be used to compute measures of complexity, as defined in \cite{preprint2}, 
\cite{amaripreprint}, and in Section 6.9 of \cite{amaribook}.
For simplicity, consider two inputs $X_1,X_2$, two outputs $Y_1,Y_2$ and a generic channel between them. 
In general, any sort of interaction is possible, which in terms of graphical models (see \cite{lauritzen}) can be 
represented by diagrams such as those in Figure \ref{fig:graphs1}.

\begin{figure}
\begin{subfigure}{.5\textwidth}
  \centering
  \includegraphics{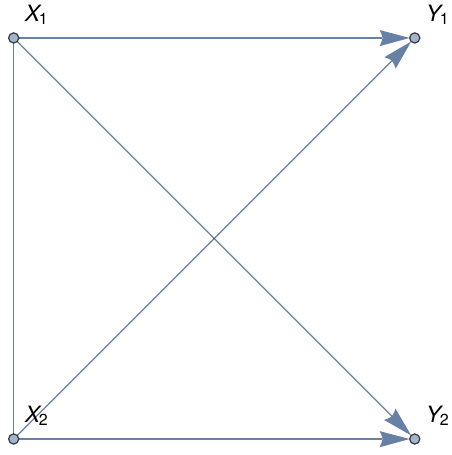}
  \label{fig:graph2}
  \caption{}
\end{subfigure}
\begin{subfigure}{.45\textwidth}
  \centering
  \includegraphics{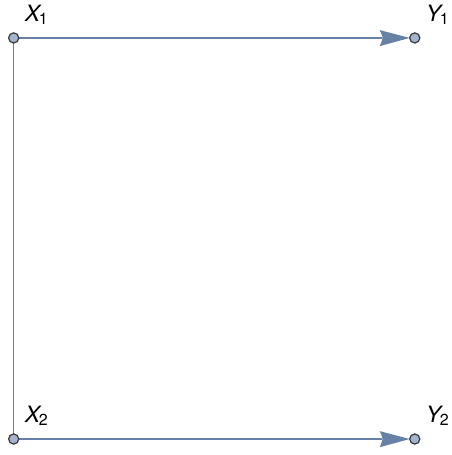}
  \label{fig:gp1}
  \caption{}
\end{subfigure}%

\caption{a) The graphical model corresponding to conditionally independent outputs $Y_1$ and $Y_2$ are 
 indeed correlated, but only indirectly, via the inputs.
 b) The graphical model corresponding to a non-complex system. 
 }
\label{fig:graphs1}
\end{figure}

Any line in the graph indicates an interaction between the nodes. 
In \cite{preprint2} the outputs are assumed to be conditionally independent, i.e. they do not directly interact (or, their
interaction can be \emph{explained away} by conditioning on the inputs). In this case the 
graph looks like Figure \ref{fig:graphs1}a, and the maginals to preserve are those of the family of pairs $(X_{I_i},Y_{J_i})$, $i=1,2$ with:
$X_{I_1}=X_{\{1,2\}}$, $Y_{J_1}=Y_{\{1\}}$, $X_{I_2}=X_{\{1,2\}}$, $Y_{J_2}=Y_{\{2\}}$.

Suppose now that $Y_1,Y_2$ correspond to $X_1,X_2$ at a later time. In this case it is natural to assume that the system
is not complex if $Y_1$ does not depend (directly) on $X_2$, and $Y_2$ does not depend (directly) on $X_1$.
Intuitively, in this case ``the whole is exactly the sum of its parts''.
In terms of graphical models, this means that our system is represented by Figure \ref{fig:graphs1}b,
meaning that the subsets of nodes in question are now only the ones given by $X_{I_1}=X_{\{1\}}$, $Y_{J_1}=Y_{\{1\}}$, $X_{I_2}=X_{\{2\}}$, $Y_{J_2}=Y_{\{2\}}$. These channels (or joints)
form an exponential family (see \cite{preprint2}) which we call $\f_1$.

\begin{figure}
\begin{subfigure}{.5\textwidth}
  \centering
  \includegraphics{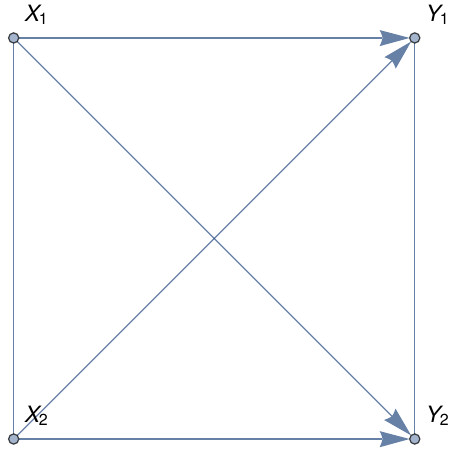}
  \label{fig:graph32}
  \caption{}
\end{subfigure}%
\begin{subfigure}{.45\textwidth}
  \centering
  \includegraphics{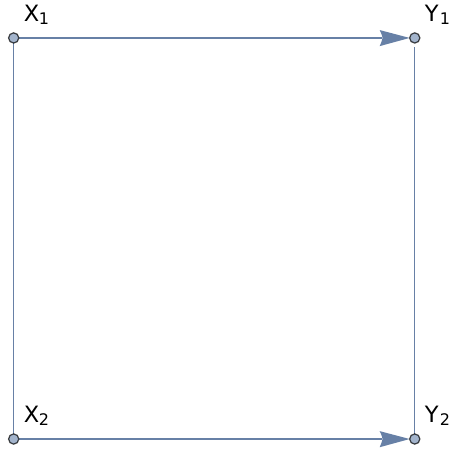}
  \label{fig:gp2}
  \caption{}
\end{subfigure}
\caption{a) The graphical model of Figure \ref{fig:graphs1}a, with correlation between the outputs.
b) The non-complex model of Figure \ref{fig:graphs1}b, with correlation between the outputs. }
\label{fig:graphs3}
\end{figure}

Suppose now, though, that the outputs are not conditionally independent anymore, because of some ``noise''
(see \cite{amaripreprint} and \cite{amaribook}).
This way the interaction structure would look like Figure \ref{fig:graphs3}, i.e. the ``complete'' subset given by $(X_I,Y_J)$ with $X_I=X_{\{1,2\}}$ and $Y_J=Y_{\{1,2\}}$.
In particular, a non-complex but ``noisy'' system would be represented by Figure \ref{fig:graphs3}b, and have 
subsets of nodes given by the pairs $(X_{I_i},Y_{J_i})$, $i=1,2,3$ with:
$X_{I_1}=X_{\{1\}}$, $Y_{J_1}=Y_{\{1\}}$, $X_{I_2}=X_{\{2\}}$, $Y_{J_2}=Y_{\{2\}}$, $X_{I_3}=X_\emptyset$, $Y_{J_3}=Y_{\{1,2\}}$.
Such channels form again an exponential family, which we call $\f_2$.

We would like now to have a measure of complexity for a channel (or joint). 
In \cite{preprint2}, the measure of complexity is defined as the divergence from the family $\f_1$ represented in Figure 
\ref{fig:graphs1}b. We will call such a measure $c_1$. In case of noise, however, it is argued in \cite{amaripreprint} and \cite{amaribook} 
that the divergence should be computed from the family $\f_2$ represented in \ref{fig:graphs3}b (for example, as
written in the cited papers, because such a complexity measure should be required to be upper bounded by
the mutual information between $X$ and $Y$).
We will call such a measure $c_2$.

Both divergences can be computed with our algorithm. As an example, we have considered the following channel:
\begin{equation}
 k(x_1,x_2,x_3;y_1,y_2) = \dfrac{1}{Z(x)}\,\exp\bigg(\big( \alpha\,x_1\,x_2 + \beta x_3\big) (y_1-y_2) \bigg)\;,
\end{equation}
with: 
\begin{equation}
 Z(x) = \sum_{y'_1,y'_2} \exp\bigg(\big( \alpha\,x_1\,x_2 + \beta x_3\big) (y'_1-y'_2) \bigg)\;.
\end{equation}
Here $X_3$ represents a node of ``unknown input noise'' that adds correlation between the outputs (of unknown form)
when if it is not observed. 
We have chosen $\alpha=1$ and $\beta=2$, and a uniform input probability $p$. 
After marginalizing out $X_3$ (obtaining then an element of the type of Figure \ref{fig:graphs3}a), we can compute the two divergences:
\begin{itemize}
 \item $c_1(k)=D_p(k||\f_1)=0.519$. 
 \item $c_2(k)=D_p(k||\f_2)=0.110$. 
\end{itemize}
This could indicate that $c_1$ is incorporating part the correlation of the output nodes due to the ``noise'', and therefore probably overestimating 
the complexity, at least in this case. 

One could nevertheless also argue that $c_2$ can underestimate complexity, as we can see in the 
following ``control'' example. Consider the channel:
\begin{equation}
 h(x_1,x_2;y_1,y_2) = \dfrac{1}{Z(x)}\,\exp\bigg(\big( \alpha\,x_1\,x_2 \big) (y_1-y_2) \bigg)\;,
\end{equation}
with: 
\begin{equation}
 Z(x) = \sum_{y'_1,y'_2} \exp\bigg(\big( \alpha\,x_1\,x_2 \big) (y'_1-y'_2) \bigg)\;,
\end{equation}
which is represented by the graph in Figure \ref{fig:graphs1}a. If the difference between $c_1$
and $c_2$ were just due to the noise, then for our new channel $c_1(h)$ and $c_2(h)$ should be equal. This is
not the case:
\begin{itemize}
 \item $c_1(h)=D_p(h||\f_1)=0.946$. 
 \item $c_1(h)=D_p(h||\f_2)=0.687$. 
\end{itemize}
The divergences are still different. This means that there is an element $h_2$ in $\f_2$, 
which does \emph{not} lie in $\f_1$, for which:
\begin{equation}
 D_p(h||h_2) < D_p(h||h_1)\quad \forall h_1\in f_1\;.
\end{equation}

The difference is this time smaller, which could mean that noise still does play 
a role, but in rigor it is hard to say, since none of these quantities is linear, and divergences 
do not satisfy a triangular inequality.

We do not want to argue here in favor or against any of these measures. We would rather like to point out that 
such considerations can be done mostly after explicit computations, which can be done with iterative scaling.

\cleardoublepage

\cleardoublepage
\appendix
\section{Proofs}

\begin{proof}[Proof of Theorem \ref{dualmk}]
 $1\Leftrightarrow2$: 
 Choose a basis $f_1,\dots,f_d$ of $\lin$. Define the map $\theta\mapsto k_\theta$, with:
 \begin{equation}
  k_\theta(x;y)=k(\theta_1,\dots,\theta_d)(x;y) := \dfrac{1}{Z_\theta(x)}\, k_0(x;y)\,\exp\left( \sum_{j=1}^d \theta_j f_j(x,y)\right)\;,
 \end{equation}
 and:
 \begin{equation}
  Z_\theta(x) := \sum_{y} k_0(x;y)\,\exp\left( \sum_{i=1}^d \theta_i f_i(x,y)\right)\;.
 \end{equation}
 Then:
 \begin{equation}
  D_p(k||k_\theta) = D_p(k||k_0) -\sum_{j=1}^d \theta_j\,\E_{pk}[f_j] + \E_p[\log Z_\theta]\;.
 \end{equation}
 Deriving (where $\D_j$ is w.r.t. $\theta_j$):
 \begin{equation}\label{pder}
  \D_jD_p(k||k_\theta) = - \E_{pk}[f_j] + \E_p\left[ \dfrac{\D_j Z_\theta}{Z_\theta} \right]\;.
 \end{equation}
 The term in the last brackets is equal to:
 \begin{align}
  \dfrac{\D_j Z_\theta}{Z_\theta} &= \dfrac{1}{Z_\theta}\,\sum_{y} k_0(x;y)\,\exp\left( \sum_{i=1}^d \theta_i f_i(x,y)\right) f_j(x,y)\\
   &= \sum_y k_\theta(x;y) f_j(x,y) \;,
 \end{align}
 so that \eqref{pder} now reads:
 \begin{equation}\label{extremum}
  \D_jD_p(k||k_\theta) = - \E_{pk}[f_j] + \E_{pk_\theta}[f_j]\;. 
 \end{equation}
 This quantity is equal to zero for every $j$ if and only if $k_\theta\in\m$.
 Now if $k_\theta$ is a minimizer, it satisfies \eqref{extremum}, and so $k_\theta\in\m$.
 Viceversa, suppose $k_\theta\in\m$, so that it satisfies \eqref{extremum} for every $j$. To prove that it is a global 
 minimizer, we look at the Hessian:
 \begin{equation}
  \D_i\D_jD_p(k||k_\theta) = \D_i\D_jD(pk||pk_\theta)\;.
 \end{equation}
 This is precisely the covariance matrix of the joint probability measure $pk_\theta$, which is positive definite.
 
 $1\Leftrightarrow3$: 
 For every $m\in\m$, we have:
 \begin{align}\label{mk0}
  D_p(m||k_0) &= \sum_{x,y}p(x)\,m(x;y)\log \dfrac{m(x;y)}{k_0(x;y)} = \E_{pm}\left[ \log \dfrac{m}{k_0} \right]\;.
 \end{align}
 If $k'\in\e$, then:
 \begin{equation}\label{mk1}
  D_p(m||k_0) = \E_{pm}\left[ \log \dfrac{m}{k'}+\log\dfrac{k'}{k_0} \right] = D_p(m||k') + \E_{pm}\left[ \log\dfrac{k'}{k_0} \right]\;.
 \end{equation}
 By definition of $\e$, the logarithm in the last brackets belongs to $\lin$, and since $m\in\m$:
 \begin{equation}
  \E_{pm}\left[ \log\dfrac{k'}{k_0} \right] = \E_{pk}\left[ \log\dfrac{k'}{k_0} \right] = \E_{pk'}\left[ \log\dfrac{k'}{k_0} \right]\;.
 \end{equation}
 Inserting in \eqref{mk1}:
 \begin{equation}\label{mk2}
  D_p(m||k_0) =  D_p(m||k') + \E_{pk'}\left[ \log\dfrac{k'}{k_0} \right] = D_p(m||k') + D_p(k'||k_0)\;.
 \end{equation}
 Since $D_p(m||k')\geq0$, \eqref{mk2} shows that $k'$ is a minimizer.
 Since $D_p(m||k_0)=D(pm||pk_0)$ is strictly convex in the first argument, its minimizer is unique.
\end{proof}

\begin{proof}[Proof of Proposition \ref{jeq}]
 \begin{align}
  p(x_I)k(x_I;y_J) &=  p(x_I)\sum_{x_{I^c},y_{J^c}}p(x_{I^c}|x_I)\,k(x_I,x_{I^c};y_J,y_{J^c}) \\
   &= \sum_{x_{I^c},y_{J^c}}p(x_I)\,p(x_{I^c}|x_I)\,k(x_I,x_{I^c};y_J,y_{J^c}) \\
   &= \sum_{x_{I^c},y_{J^c}}p(x_I,x_{I^c})\,k(x_I,x_{I^c};y_J,y_{J^c}) \\
   &= \sum_{x_{I^c},y_{J^c}}pk(x_I,x_{I^c},y_J,y_{J^c}) \\
   &= pk(x_I,y_J)\;.
 \end{align}
\end{proof}

\begin{proof}[Proof of Proposition \ref{presc}]
 For $f$ in $F_{IJ}$:
 \begin{equation}\label{mtwo1}
   \E_{pk}[f] = \sum_{x,y} p(x)\,k(x;y)\,f(x,y) = \sum_{x_I,y_J} p(x_I)\,k(x_I;y_J)\,f(x_I;y_J)\;,
 \end{equation}
 and just as well:
 \begin{equation}\label{mtwo2}
   \E_{p\bar{k}}[f] = \sum_{x,y} p(x)\,\bar{k}(x;y)\,f(x,y) = \sum_{x_I,y_J} p(x_I)\,\bar{k}(x_I;y_J)\,f(x_I;y_J)\;.
 \end{equation}
 The definition in \eqref{m} (with strict positivity of $p$) requires exactly that:
 \begin{equation}
  \E_{pk}[f]=\E_{p\bar{k}}[f]
 \end{equation}
 for every $f\in F_{IJ}$. Using \eqref{mtwo1} and \eqref{mtwo2}, the equality becomes:
 \begin{equation}
  \sum_{x_I,y_J} p(x_I)\,k(x_I;y_J)\,f(x_I;y_J) = \sum_{x_I,y_J} p(x_I)\,\bar{k}(x_I;y_J)\,f(x_I;y_J)
 \end{equation}
 for every $f$ in $F_{IJ}$, which means that $k(x_I;y_J)=\bar{k}(x_I;y_J)$.
\end{proof}

\begin{proof}[Proof of Proposition \ref{jlevel}]
 We have:
 \begin{align}
  p \s_{IJ}^{\bar k} k\,(x_I,y_J) &= p \s_{IJ}^{\bar k} k\,(x_I,x_{I^c},y_J,y_{J^c})\\
   &= p(x_I,x_{I^c}) k(x_I , x_{I^c} ; y_J , y_{J^c} ) \, \frac{ \bar{k}(x_I ; y_J) }{ k(x_I ; y_J) } \\
   &= pk(x_I , x_{I^c} , y_J , y_{J^c} ) \, \frac{p(x_I) \, \bar{k}(x_I ; y_J) }{p(x_I) \, k(x_I ; y_J) } \\
   &= pk(x_I , x_{I^c} , y_J , y_{J^c} ) \, \frac{p\bar{k}(x_I , y_J) }{pk(x_I , y_J) } \\
   &= \sigma_{IJ}^{p\bar k} pk\,(x,y)\;.
 \end{align}
\end{proof}

\begin{proof}[Proof of Proposition \ref{inscale}]
 The first member is equal to:
 \begin{align}
  p(x) N \s_{IJ}^{\bar k} k(x,y) &= p(x)\, \dfrac{\s_{IJ} k \, (x,y)}{\sum_{y'}\s_{IJ}^{\bar k} k \, (x,y')} \\
   &= p(x)\, \dfrac{\s_{IJ}^{p \bar k} pk (x)\,  \s_{IJ} k \, (x,y)}{\sum_{y'}\s_{IJ}^{p \bar k} pk (x)\,\s_{IJ}^{\bar k} k \, (x,y')} \\
   &= p(x)\, \dfrac{\s_{IJ}^{p \bar k} pk\,(x,y)}{\sum_{y'}\s_{IJ}^{p \bar k} pk \, (x,y')} \\
   &= \s_{IJ}^{p \bar k} pk\,(x,y)\,\dfrac{p(x)}{\s_{IJ}^{p \bar k} pk \, (x)} \\
   &=  \s_{[N]}^{p} \s_{IJ}^{p \bar k} pk\;.
 \end{align}
\end{proof}

\begin{proof}[Proof of Theorem \ref{convergence}]
 In the hypothesis and notation of Theorem \ref{jointit}, take the collection $\lin_1,\dots,\lin_{2n}$ in the following way,
 for $i=1,\dots,n$:
 \begin{itemize}
  \item $\lin_{2i-1}:=\mathcal{J}_{[N]}(p)$;
  \item $\lin_{2i}:=\mathcal{J}_{I_iJ_i} (p \bar k)$.
 \end{itemize}
 Their intersection $\lin$ is nonempty, as it contains at least $p\bar k$.
 
 Take as initial distribution $q^0:=pk^0$ and form as in Theorem \ref{jointit} the sequence $\{q^j\}$ of $I$-projections. 
 According to Theorem \ref{jointit}, this sequence converges to the $I$-projection of $pk^0$ on $\lin$.
 Since $\lin \subseteq \mathcal{J}_{[N]}(p)$, this projection will have input marginal equal to $p(X)$, and so we can write 
 it as $p(X)\,l(X;Y)$ for some uniquely defined channel $l$. We have, for $j\to\infty$:
 \begin{equation}
  q^j \to pl\;,
 \end{equation}
 so in particular, for the subsequence of even-numbered terms also:
 \begin{equation}
  q^{2j} \to pl\;.
 \end{equation}
 This subsequence is defined iteratively by:
 \begin{equation}
  q^{2j} = \s_{[N]}^{p} \s_{I_jJ_J}^{p \bar k} q^{2(j-1)}\;.
 \end{equation}
 Propositions \ref{jlevel} and \ref{inscale} imply then that:
 \begin{equation}
  q^{2j} = p\,k^j
 \end{equation}
 for every $j$, where $\{k^j\}$ is the sequence defined in the statement of Theorem \ref{convergence}.
 Therefore this sequence converges:
 \begin{equation}
  k^j \to l\;.
 \end{equation}
 
 Since $pl\in\lin\subseteq \lin_i$ for all $i$, $l\in \m_i$ for all $i$ because of \eqref{corresp}, which by 
 definition means that $l\in \m$.
 Moreover, $pl$ is the $I$-projection $q$ of $pk^0$ on $\lin$, which means that:
  \begin{equation}
  D(pk^0||p\bar k) =  D(pk^0||pl) + D(pl||p\bar k)\;.
 \end{equation}
 Using the chain rule of the KL-divergence \eqref{pkdiv}, we get:
 \begin{equation}
  D_p(k^0||\bar k) = D_p(k^0||l) + D_p(l||\bar k)\;,
 \end{equation}
 which means that $l$ is the $I$-projection of $k^0$ on $\m$.
\end{proof}

\end{document}